%% file: main.tex
\documentclass[12pt]{article}
\usepackage[papersize={8.5in,12in},margin=1in]{geometry}
\usepackage{graphicx}
\usepackage[T5,T1]{fontenc}
\usepackage[utf8]{inputenc}
\usepackage{amsmath}
\usepackage{amsthm}
\usepackage{amssymb}
\usepackage{amsfonts}
\usepackage{authblk}
\usepackage{mathtools}
\usepackage[ruled, vlined, linesnumbered]{algorithm2e}
\usepackage{algpseudocode}
\usepackage{thmtools}
\usepackage{thm-restate}
\usepackage{tikz}
\usepackage{booktabs}
\usepackage{multirow}
\usepackage{tabularx}
\usepackage{subcaption}
\usepackage{xcolor}
\usepackage{hyperref}
\hypersetup{
    colorlinks=true,
    linkcolor=teal,
    filecolor=blue,      
    urlcolor=blue,
    citecolor=teal,
}
\usepackage{cleveref}
\usetikzlibrary{arrows.meta, positioning, calc, backgrounds, fit}

\newcommand{\problemStatement}[3]{%
  \begin{center}
  \begin{tabularx}{\columnwidth}{@{}lX@{}}
  \toprule
  \multicolumn{2}{@{}c@{}}{#1}\tabularnewline
  \midrule
  \bfseries Instance:    & #2 \\
  \bfseries Question: & #3 \\
  \bottomrule
  \end{tabularx}
  \end{center}
}

\newtheorem{theorem}{Theorem}
\newtheorem{lemma}{Lemma}
\newtheorem{observation}{Observation}
\newtheorem{corollary}{Corollary}
\newtheorem{claim}{Claim}
\newtheorem{question}{Question}

\theoremstyle{definition}

\newtheorem{definition}{Definition}

\title{\textbf{Independent Set Reconfiguration\\
via\\ Dilworth Decompositions}}

\author{Ziad Ismaili Alaoui
\thanks{\hspace*{-1.8em}School of Computer Science and Informatics,
University of Liverpool, United Kingdom\\
Department of Informatics,
Philipps-Universität Marburg, Germany\\
\texttt{ziad.ismaili-alaoui@liverpool.ac.uk}\\
\url{https://www.ziadism.com}}}
\date{}

\makeatletter
\def\thanks#1{\protected@xdef\@thanks{\@thanks
        \protect\footnotetext{#1}}}
\makeatother

\begin{document}

\maketitle

\begin{abstract}
\noindent The \textsc{Token Jumping} and \textsc{Sliding Token} problems are fundamental reconfiguration problems defined on the independent sets of an undirected graph. Given two independent sets $I$ and $J$, each of size $k$, these problems ask whether there exists a sequence of elementary operations transforming $I$ into $J$ such that every intermediate configuration is also an independent set of size $k$. Suppose a token is placed on each vertex of $I$: in \textsc{Sliding Token}, an operation moves a token from a vertex $u \in I$ to an adjacent vertex $v \notin I$; in \textsc{Token Jumping}, the token may instead move to any vertex $v \notin I$. While both problems are $\mathsf{PSPACE}$-complete on general graphs, polynomial-time algorithms for one or both variants have been developed for several graph classes, including trees, block graphs, bipartite permutation graphs, cographs, $P_4$-tidy graphs, and interval graphs.

In this paper, we prove that both problems are solvable in polynomial time on threshold signed graphs, also known as Dilworth-2 graphs. These are precisely the graphs whose vertex sets can be partitioned into at most two chains under the domination preorder $u \preceq v$ defined by $N(u)\subseteq N[v]$. Our approach exploits this chain structure. For \textsc{Token Jumping}, it extends to an arbitrary number of chains: we obtain an $n^{O(\mathcal{D}(G))}$-time algorithm, where the Dilworth number $\mathcal{D}(G)$ is the minimum number of chains in such a partition. Thus, \textsc{Token Jumping} belongs to $\mathsf{XP}$ when parameterised by the Dilworth number.

This graph class is a subclass of permutation graphs, for which the complexity of these problems remains open, and is incomparable with the class of bipartite permutation graphs studied by Fox-Epstein et al.~(ISAAC, 2015). Our algorithms are based on the inclusion-chain structure that characterises threshold signed graphs, an algorithmic approach that may be useful beyond this graph class.\newline

\end{abstract}

\section{Introduction}
\label{s:intro}
\input{intro}

\section{Preliminaries}
\label{s:prelims}
\input{prelims}

\section{On \textsc{Token Jumping}}
\label{s:tj}
\input{tj}

\section{On \textsc{Sliding Token}}
\label{s:st}
\input{st}

\section{Concluding Remarks}
\label{s:conclusion}
\input{conclusion}

\bibliographystyle{alpha}
\bibliography{bib}

\end{document}

%% file: intro.tex
Reconfiguration problems typically ask whether one valid solution can be transformed into another via a sequence of elementary steps (see~\cite{nishimura2018introduction} for an introduction). A good example is \textsc{Independent Set Reconfiguration}, which is surprisingly simple to formulate, yet hides rich combinatorial structure. Among its most widely studied variants are \textsc{Token Jumping} and \textsc{Sliding Token}. Imagine placing tokens on the vertices of a graph so that they form an independent set $I$, and then moving one token at a time (along an edge in \textsc{Sliding Token}, or to any vertex in \textsc{Token Jumping}) until reaching another independent set $J$ of the same size. Throughout the process, the occupied vertices must remain an independent set; i.e.~two tokens cannot be adjacent. The question then is: Can all tokens on $I$ be placed onto $J$ in this way? More formally, the problem is defined as follows.

\problemStatement{\textsc{Independent Set Reconfiguration} (\textsc{Token Jumping} and \textsc{Sliding Token})}
  {An undirected graph $G$ and independent sets\footnote{We use the terms \emph{independent set} and \emph{configuration} interchangeably.} $I,J \subseteq V(G)$, with $|I|=|J|=k$.}
  {Does there exist a finite sequence of independent sets $I_1,\ldots,I_m$ such that
$I_1=I$, $I_m=J$, $|I_i|=k$ for all $i\in[m]$, $|I_i\triangle I_{i+1}|=2$ for all $i\in[m-1]$, and, in the case of \textsc{Sliding Token}, $I_i\triangle I_{i+1}=\{u,v\}\in E(G)$ for all $i\in[m-1]$?}

\noindent Although both variants are $\mathsf{PSPACE}$-complete when we know nothing about the input graphs \cite{hearn2005pspace,ito2011complexity}, only a handful of graph classes are known to have enough structure to admit polynomial-time algorithms. Since the introduction of \textsc{Sliding Token} by Hearn and Demaine~\cite{hearn2005pspace} and of \textsc{Token Jumping} (seemingly independently) by Ito et al.~\cite{ito2011complexity} and Kamiński et al.~\cite{kaminski2012complexity}, a series of results has gradually expanded this list of tractable classes to include trees, interval graphs, cographs, and several others (see Table~\ref{tab:st-complexity} for an overview\footnote{For cacti under \textsc{Sliding Token}, it was claimed in~\cite{anh2016sliding} they were solvable in polynomial time, but a later note by the authors revealed that the algorithm was incorrect~\cite{duc24cactus}.}). Nevertheless, positive results remain remarkably scarce. As Bartier, Bousquet, and Mouawad recently observed in~\cite{bartier2023galactic}, ``almost no positive result is known for \textsc{Token Sliding} [\textit{sic}]\footnote{The exact name of the problem is not standardised in the literature. Variants such as \textsc{Sliding Token}, \textsc{Sliding Tokens}, and \textsc{Token Sliding} all appear. We arbitrarily follow the terminology of \cite{demaine2014polynomial}, despite the other reconfiguration problem being named \textsc{Token Jumping} in our paper.} even for incredibly simple cases like bounded-degree graphs.''

\begin{table}[h]
\centering
\renewcommand{\arraystretch}{1.2}
\begin{tabular}{l|l|l}
\toprule
\textbf{Graph Class} & \textbf{\textsc{Jumping}} & \textbf{\textsc{Sliding}} \\
\midrule

Arbitrary & $\mathsf{PSPACE}$-c. \cite{ito2011complexity} & $\mathsf{PSPACE}$-c. \cite{hearn2005pspace} \\
Planar & $\mathsf{PSPACE}$-c. \cite{ito2011complexity,kaminski2012complexity} & $\mathsf{PSPACE}$-c. \cite{hearn2005pspace} \\
Bounded bandwidth & $\mathsf{PSPACE}$-c. \cite{wrochna2018reconfiguration} & $\mathsf{PSPACE}$-c. \cite{wrochna2018reconfiguration} \\
Bipartite & $\mathsf{NP}$-complete \cite{lokshtanov2018complexity} & $\mathsf{PSPACE}$-c. \cite{lokshtanov2018complexity} \\

$C_{2k\geq 4}$-free& $\mathsf{P}$ \cite{kaminski2012complexity} & $\mathsf{PSPACE}$-c. \cite{belmonte2021token} \\
Split & $\mathsf{P}$ ($C_{2k\geq 4}$-free) & $\mathsf{PSPACE}$-c. \cite{belmonte2021token} \\
Trees & $\mathsf{P}$ ($C_{2k\geq 4}$-free) & $\mathsf{P}$ \cite{demaine2014polynomial} \\
Chordal & $\mathsf{P}$ ($C_{2k\geq 4}$-free) & $\mathsf{PSPACE}$-c. \cite{belmonte2021token} \\
Claw-free & $\mathsf{P}$ \cite{bonsma2014reconfiguring} & $\mathsf{P}$ \cite{bonsma2014reconfiguring} \\
Interval & $\mathsf{P}$ ($C_{2k\geq 4}$-free) & $\mathsf{P}$ \cite{bonamy2017token} \\
$P_4$-tidy & Open & $\mathsf{P}$ \cite{busolini2026token} \\
Cacti & Open & Open \\
Block & $\mathsf{P}$ (chordal) & $\mathsf{P}$ \cite{francis2025token} \\

\midrule

\multirow{3}{*}{Permutation} 
& General: Open & General: Open \\
& Cographs: $\mathsf{P}$ \cite{bonsma2016independent} & Cographs: $\mathsf{P}$ \cite{kaminski2012complexity} \\
& Bipartite: Open & Bipartite: $\mathsf{P}$ \cite{fox2015sliding} \\
& Threshold signed: $\mathsf{P}$ [Thm~\ref{thm:tj}] & Threshold signed: $\mathsf{P}$ [Thm~\ref{thm:st}] \\

\bottomrule
\end{tabular}
\caption{Known complexities for \textsc{Independent Set Reconfiguration} across various graph classes (not exhaustive). Our results contribute towards the open case of permutation graphs.} \label{tab:st-complexity} 
\end{table}

\noindent Beyond the theoretical appeal of this problem, some positive results appear to have direct applications in industrial fields such as robotics; for example, see \cite{11626826}, where Zucker and Ben-Abu show \textsc{Sliding Token} is solvable in linear time on Ptolemaic (gem-free chordal) graphs, which they relate to the hard-core lattice gas model in statistical mechanics.

\paragraph{Our contributions.} In this paper, we further contribute to the complexity landscape by proving that both \textsc{Token Jumping} and \textsc{Sliding Token} are solvable in polynomial time on threshold signed graphs (also known as Dilworth-2 graphs). A graph $G=(V,E)$ is a \emph{threshold signed graph} if there exist a mapping $a:V\to\mathbb{R}$ and positive real constants $S$ and $T$ such that $|a(v)| < \min\{S,T\}$ for all $v \in V$ and, for any distinct vertices $u,v\in V$, $\{u,v\}\in E$ if and only if $|a(u)+a(v)|\ge S$ or $|a(u)-a(v)|\ge T$. Threshold signed graphs have been extensively studied in the literature since the late 1970s (see, for example, \cite{foldes1977split,benzaken1985threshold,calamoneri2014pairwise}) and are recognisable in linear time \cite{benzaken1985split}. We see two reasons why threshold signed graphs are a natural class to study. 

First, Benzaken et al.~\cite{benzaken1985threshold} showed that they form a strict subclass of permutation graphs. Our result then identifies a new tractable non-trivial subclass, while the complexity of both problems on general permutation graphs remains open.\footnote{The problems are known to be tractable on permutation graphs when $k=\alpha(G)$, where $\alpha(G)$ denotes the size of a maximum independent set. Aoike et al.~showed this for \textsc{Sliding Token}~\cite{aoike2024finding}; \textsc{Token Jumping} follows from their arguments.} Moreover, while \textsc{Sliding Token} is known to be solvable in polynomial time on bipartite permutation graphs~\cite{fox2015sliding}, they and threshold signed graphs are incomparable: neither contains the other (apart from $K_2$ and $C_4$, no non-trivial ladder graph is threshold signed\footnote{Threshold signed graphs are $(C_4 \cup P_2)$-free, which ladder graphs are generally not.}).

And second, threshold signed graphs are precisely the Dilworth-2 graphs, that is, graphs whose vertex set can be partitioned into two chains ordered by neighbourhood inclusion. Our algorithms rely almost entirely on this two-chain structure rather than on the numerical threshold representation. Dilworth-1 graphs (threshold graphs) are cographs and hence already well understood from the perspective of reconfiguration~\cite{bonsma2016independent}, which makes Dilworth-2 a natural next step.

To the best of our knowledge, this is the first work on \textsc{Independent Set Reconfiguration} that systematically exploits the Dilworth decomposition of the input graph. We believe this technique may be useful for studying reconfiguration on other graph classes of bounded Dilworth number.

\input{fig-intro}

Both of our algorithms are based on the same core idea: Given two independent sets $I$ and $J$, each algorithm attempts to transform them into a so-called \emph{canonical} configuration, where two mutually reachable configurations have the same canonical configuration. The resulting canonical configurations are then compared: if they are equal, then there exists a sequence of slides or jumps transforming $I$ into $J$; otherwise, no such sequence exists. This is justified by the fact that every valid reconfiguration sequence can be reversed. This approach has appeared repeatedly in the literature on \textsc{Independent Set Reconfiguration}. For example, Demaine et al.~\cite{demaine2014polynomial} showed that every movable configuration on a tree can be transformed into a single canonical configuration. Their algorithm therefore reduces to identifying the tokens that are \emph{rigid}, that is, unable to move in any valid reconfiguration sequence. Other examples include the algorithms for bipartite permutation graphs \cite{fox2015sliding}, block graphs \cite{francis2025token}, and Ptolemaic graphs~\cite{11626826}.

To give a flavour of our algorithm for \textsc{Token Jumping}, suppose the input graph is partitioned into two neighbourhood-inclusion chains $A$ and $B$, as guaranteed by the Dilworth-2 property. Starting from an independent set $I$, we first greedily ``flush'' all tokens on each chain as far to the left as possible, where the neighbourhood of each vertex is contained in the neighbourhood of any vertex to its right. We then repeatedly move the rightmost token on chain $B$ to the leftmost available vertex on chain $A$, whenever possible, until no further such move exists. The resulting independent set is the canonical configuration associated with $I$. We then repeat the same procedure from $J$ to obtain another canonical configuration, and reachability between $I$ and $J$ is simply determined by comparing the two: they coincide if and only if $I$ and $J$ are reachable from one another. In fact, this comparison reduces to checking whether the number of tokens in $A$ (or $B$ by symmetry) is the same in both canonical configurations. We therefore obtain the following result.

\begin{restatable}{theorem}{tjmain}
\label{thm:tj}
There exists a polynomial-time algorithm to decide \textsc{Token Jumping} on threshold signed graphs.
\end{restatable}

\noindent Beyond threshold signed graphs, the same structural
observations give an $\mathsf{XP}$ algorithm for \textsc{Token Jumping}
parameterised by the Dilworth number. As a consequence, this gives
a polynomial-time algorithm on Dilworth-$3$ graphs and, more
generally, on every graph class of bounded Dilworth number.

\begin{restatable}{theorem}{xpmain}
\label{thm:xp}
There exists an $n^{O(d)}$-time algorithm to decide \textsc{Token Jumping} on arbitrary $n$-vertex graphs of Dilworth number $d$.
\end{restatable}

\noindent As one would expect, the situation for \textsc{Sliding Token} is a tad more delicate, as tokens cannot just go anywhere and must slide along incident edges. In particular, the ``flushing'' procedure as described for \textsc{Token Jumping} clearly does not work. Fortunately, the same high-level idea still holds. We show that, at any point, only a small set of tokens needs to be considered for sliding, which gives us control over how configurations evolve. This leads to a simple greedy procedure that first slides tokens from $A$ to $B$, and then from $B$ back to $A$, until a canonical configuration is reached after a polynomial number of steps. This canonical configuration is determined purely by the exact positions of the tokens. We obtain a second canonical configuration by repeating the same procedure from $J$, and, again, reachability is decided by comparing the two.

\begin{restatable}{theorem}{stmain}
\label{thm:st}
There exists a polynomial-time algorithm to decide \textsc{Sliding Token} on threshold signed graphs.
\end{restatable}

\noindent Beyond the specific polynomial-time algorithms presented here, we view the main conceptual contribution of this work as showing that the Dilworth decomposition of a graph into neighbourhood-inclusion chains can actually serve as a useful algorithmic tool for \textsc{Independent Set Reconfiguration}. We hope that this perspective will prove useful beyond threshold signed graphs, particularly for graph classes of bounded Dilworth number for \textsc{Sliding Token} and, more generally, in understanding the algorithmic role of neighbourhood-inclusion decompositions in graph reconfiguration in general.

\paragraph{Outline.} Our paper is organised as follows. Section~\ref{s:prelims} introduces the notation and definitions used throughout the paper. Sections~\ref{s:tj} and~\ref{s:st} present our algorithms for \textsc{Token Jumping} and \textsc{Sliding Token}, respectively. Finally, we conclude in Section~\ref{s:conclusion}.

%% file: fig-intro.tex
\begin{figure}[h]
\centering
\begin{tikzpicture}[
  vertex/.style={circle, fill=black, inner sep=0pt, minimum size=4pt},
  every path/.style={thin, black}
]
 
\node[vertex] (A) at (0,2) {};
\node[vertex] (B) at (1.6,2) {};
\node[vertex] (C) at (3.2,2) {};
\node[vertex] (D) at (4.8,2) {};
\node[vertex] (E) at (6.4,2) {};
\node[vertex] (F) at (8.0,2) {};
 
\node[vertex] (G) at (0,0) {};
\node[vertex] (H) at (1.6,0) {};
\node[vertex] (I) at (3.2,0) {};
\node[vertex] (J) at (4.8,0) {};
\node[vertex] (K) at (6.4,0) {};
\node[vertex] (L) at (8.0,0) {};
 
\node[above] at (A.north) {$0.8$};
\node[above] at (B.north) {$1.6$};
\node[above] at (C.north) {$2.1$};
\node[above] at (D.north) {$2.7$};
\node[above] at (E.north) {$3.2$};
\node[above] at (F.north) {$4$};
\node[below] at (G.south) {$-0.6$};
\node[below] at (H.south) {$-1.2$};
\node[below] at (I.south) {$-1.8$};
\node[below] at (J.south) {$-2.4$};
\node[below] at (K.south) {$-3.2$};
\node[below] at (L.south) {$-4$};
 
\node[align=left] at (-1.6,1) {$S=4.5$\\ $T=5$};
 
\draw (A) to[bend left=50] (F);
\draw (B) to[bend left=30] (E);
\draw (B) to[bend left=40] (F);
\draw (C) -- (D);
\draw (C) to[bend left=20] (E);
\draw (C) to[bend left=30] (F);
\draw (D) -- (E);
\draw (D) to[bend left=20] (F);
\draw (E) -- (F);
 
\draw (G) to[bend right=50] (L);
\draw (H) to[bend right=40] (L);
\draw (I) to[bend right=20] (K);
\draw (I) to[bend right=30] (L);
\draw (J) -- (K);
\draw (J) to[bend right=20] (L);
\draw (K) -- (L);
 
\draw (B) -- (L);
\draw (C) -- (K);
\draw (C) -- (L);
\draw (D) -- (J);
\draw (D) -- (K);
\draw (D) -- (L);
\draw (E) -- (I);
\draw (E) -- (J);
\draw (E) -- (K);
\draw (E) -- (L);
\draw (F) -- (H);
\draw (F) -- (I);
\draw (F) -- (J);
\draw (F) -- (K);
\draw (F) -- (L);
 
\end{tikzpicture}
\caption{A threshold signed graph with a Dilworth-2 decomposition (the top vertices form one chain, and the bottom ones, the other). In each chain, vertices are ordered left-to-right by increasing neighbourhood inclusion (dominance toward the right). Note that a Dilworth-2 graph need not be a cograph.}
\end{figure}
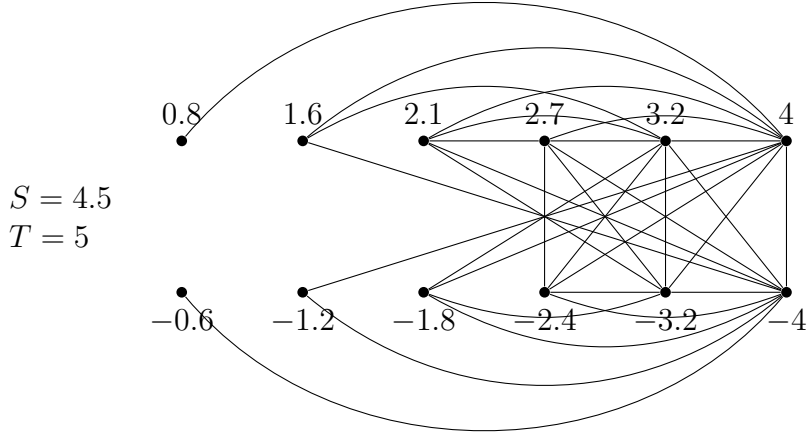

%% file: prelims.tex
\paragraph{Basic notation.} We write $[n]$ for $\{1,2,\dots,n\}$. All graphs in this paper are finite, simple, and undirected. For a graph $G=(V,E)$, we write $V(G)$ and $E(G)$ for its vertex and edge sets, respectively. We always refer to $|V(G)|$ as $n$. For a vertex $v\in V(G)$, the \emph{open neighbourhood} of $v$ is $\Gamma(v)=\{u\in V(G): \{u,v\}\in E(G)\}$, and the \emph{closed neighbourhood} is $\Gamma[v]=\Gamma(v)\cup\{v\}$. The degree of $v$ is $|\Gamma(v)|$. For sets $A,B$, we write $A\triangle B=(A\setminus B)\cup(B\setminus A)$ for their symmetric difference. We refer to the subgraph of $G$ induced on $S \subseteq V(G)$ as $G[S]=(S,E(G) \cap {S \choose 2})$.

An \emph{independent set} of $G$ is a set of pairwise non-adjacent vertices. In the context of reconfiguration, we also refer to independent sets as \emph{configurations}.

\begin{definition}[Threshold signed graph]
A graph $G=(V,E)$ is a \emph{threshold signed graph} if there exist a mapping $a:V\to\mathbb{R}$ and positive real constants $S,T>0$ such that $|a(v)|< \min\{S,T\}$, for all distinct $u,v\in V$,
\[
\{u,v\}\in E \quad \text{if and only if} \quad |a(u)+a(v)|\ge S \ \text{ or }\ |a(u)-a(v)|\ge T.
\]
\end{definition}

\noindent We next introduce a combinatorial characterisation of this class. For vertices $u,v\in V(G)$, we write $u\preceq v$ if $\Gamma(u)\subseteq \Gamma[v]$; here, we say that $v$ \emph{dominates} $u$.

\begin{definition}[Dilworth-$k$ graph, Dilworth number]
A graph $G$ is a \emph{Dilworth-$k$ graph} if $V(G)$ can be partitioned into at most $k$ chains, where a chain $(v_1,\ldots,v_t)$ satisfies
\[
v_i \preceq v_j \qquad \text{for all } 1\le i<j\le t.
\]
The minimum such $k$ is the \emph{Dilworth number} of $G$, denoted $\mathcal{D}(G)$.
\end{definition}

\begin{theorem}[Benzaken et al.~\cite{benzaken1985threshold}]
A graph is threshold signed if and only if it is a Dilworth-2 graph.
\end{theorem}

\noindent Hence, every threshold signed graph admits a partition of its vertex set into two neighbourhood-inclusion chains (note that a chain can be empty); in fact, it turns out $V(G)=\{v:a(v)<0\} \sqcup \{v:a(v) \geq 0\}$ is always a valid partition. 

\paragraph{Input.} Since such a partition can be computed in polynomial time~\cite{benzaken1985split,fulkerson1956note}, we simply assume throughout that it is given as part of the input; vertices with equal neighbourhoods in a chain can be ordered arbitrarily. We thus write $V(G)=A\sqcup B$, where $A$ and $B$ are the fixed partitions of the neighbourhood-inclusion chains. For $X\in\{A,B\}$ and $u,v\in V(G)$, we write $u<_Xv$ if $u,v\in X$ and $u$ precedes $v$ in the fixed linear order of $X$. In particular, $u<_Xv$ implies $\Gamma(u)\subseteq\Gamma[v]$.\footnote{Again, vertices that dominate each other are ordered arbitrarily, with this order fixed throughout.} 

We then assume the input graph to be labelled with respect to this order, with $A = \{a_1,a_2,\dots,a_{|A|}\}$ and $B = \{b_1,b_2,\dots,b_{|B|}\}$, where $a_1 <_A a_2 <_A \dots <_A a_{|A|}$ and $b_1 <_B b_2 <_B \dots <_B b_{|B|}$ are the fixed orders. By pure abuse of notation, we often refer to these vertex partitions as \emph{chains}.

\paragraph{Reconfiguration.} We now introduce the reconfiguration terminology. We consider tokens to be uniquely labelled. Fix $k\in [n]$. A size-$k$ independent set is interpreted as a configuration of $k$ tokens placed on vertices of $G$, with no two tokens on adjacent vertices. A \emph{reconfiguration move} either (i) moves a token to an unoccupied vertex (\textsc{Token Jumping}), or (ii) moves a token along an edge to an unoccupied neighbour (\textsc{Sliding Token}).

\begin{definition}[$k$-token reconfiguration graph]
The \emph{$k$-token reconfiguration graph} $\mathcal{R}_k(G)$ has as vertex set all independent sets of $G$ of size $k$. Two configurations $X,Y$ are adjacent in $\mathcal{R}_k(G)$ if and only if $X\triangle Y=\{u,v\}$ and moving the token from $u$ to $v$ is a valid reconfiguration move under the given rule in the context being used.
\end{definition}

\noindent Since every move is clearly reversible, observe that reachability in $\mathcal{R}_k(G)$ induces an equivalence relation on the set of size-$k$ independent sets. For configurations $I$ and $J$, we write $I\leftrightsquigarrow J$ if they lie in the same connected component of $\mathcal{R}_k(G)$, equivalently if one can be transformed into the other via a finite sequence of valid moves. The \emph{reconfiguration component} of $I$ is defined as
\[
\mathcal{C}(I)=\{J \mid I\leftrightsquigarrow J\}.
\]

\noindent In other words, $\mathcal{C}(I)$ is the set of all independent sets reachable from $I$. For a configuration $I$ and a token $t$ on $v\in I$, we say that $v$ is \emph{occupied}; in addition, $v$ is \emph{unoccupied} when $v \notin I$. The token $t$ on $v \in I$ is \emph{$I$-movable} if there exists a configuration $J$ adjacent to $I$ in $\mathcal{R}_k(G)$ such that $t$ occupies $u \not= v$.

%% file: tj.tex
This section is devoted to the proofs of Theorems~\ref{thm:tj} and~\ref{thm:xp}. Our algorithm takes advantage of the fact that neighbourhoods are linearly ordered within each chain of the Dilworth-2 partition. In particular, tokens in each chain can be moved from right to left without violating independence (that is, creating adjacencies among tokens), since neighbourhoods can only decrease along the order. This is a highly useful property since it allows us to transform any configuration into a canonical form (i.e.~a \emph{unique} representative independent set in the reconfiguration component) for both $I$ and $J$ by first greedily moving tokens within each chain, and then transferring tokens between the two chains in a controlled manner. Reachability is then simply decided by comparing the resulting canonical configurations.

\paragraph{Section-specific terminology.} Let $G=(V,E)$ be a threshold signed graph, and let $V(G)=A\sqcup B$, with $a_1 <_A a_2 <_A \dots <_A a_{|A|}$ and $b_1 <_B b_2 <_B \dots <_B b_{|B|}$, be a fixed partition into two neighbourhood-inclusion chains, as guaranteed by the Dilworth-2 property. Recall that vertices with equal neighbourhoods are ordered arbitrarily as long as that remains fixed. We always refer to the number of tokens as $k$. 

For a chain $X\in\{A,B\}$ and an independent set $I$, let $\nabla_{I}(X)=|I\cap X|$ denote the number of tokens placed on $X$. For $u,v \in X$, we say that $u$ is on the \emph{left} of $v$ when $u <_X v$; symmetrically, $v$ is on the \emph{right} of $u$. A \emph{leftward move} in a chain $X$ is a token move from an occupied vertex $v\in X$ to an unoccupied vertex $u\in X$ satisfying $u <_X v$. A \emph{transfer} from a chain $X$ to a chain $Y$ is a token move from the rightmost occupied vertex $u\in X$ to the leftmost unoccupied vertex $v \in Y$ not adjacent to a vertex in $V(G) \setminus \{u\}$ occupied by a token.

An independent set is \emph{left-normalised} if neither chain admits a leftward move of a token. A \emph{left-normalisation} of $I$ is the following procedure performed as long as possible: Find any token on a vertex $u$ in some chain $X$, and move it to an unoccupied vertex $v \in X$ such that $v <_X u$. Let $L(I)$ be the resulting configuration; clearly, $\nabla_{I}(X)=\nabla_{L(I)}(X)$ for $X \in \{A,B\}$. Let $T_{B \to A}(I)$ denote the configuration obtained after \emph{attempting} a single transfer from $B$ to $A$ from $I$; such a transfer may not be possible, in which case $I=T_{B \to A}(I)$.

\paragraph{Canonical configuration.} We define the canonical configuration $\text{can}(I)$ of $I$ as follows: \[
\text{can}(I)=\bigcup_{i=1}^{m(I)}\{a_i\} \cup \bigcup_{i=1}^{k-m(I)}\{b_i\}
\] where $m(I)= \max_{J \in \mathcal{C}(I)} \nabla_J(A)$.
In other words, $\text{can}(I)$ is the left-normalised configuration of a reachable configuration that maximises the number of tokens on the vertices of $A$.

\paragraph{Algorithm.} The algorithm is quite easy: Starting from $I$, compute $L(I)$ by greedily moving tokens onto vertices on the left of their respective chains, as long as possible (this is the ``flushing'' idea). Then, from $L(I)$, as long as possible, perform transfers from $B$ to $A$. The resulting configuration $\text{can}(I)$ is then the unique canonical configuration of $\mathcal{C}(I)$. Repeat the same procedure from $J$ to obtain $\text{can}(J)$, and return ``yes'' if and only if $\text{can}(I)=\text{can}(J)$.

\begin{algorithm}
\caption{\textsc{Token-Jumping-Reachability}}
\KwIn{Threshold signed graph $G=(V,E)$ and independent sets $I,J$.}

$I \gets L(I)$ \Comment{From $I$, left-normalise.}

\While{$T_{B \to A}(I) \not= I$}{
    $I \gets T_{B \to A}(I)$ \Comment{Transfer as many tokens from $B$ to $A$ from $L(I)$ as possible.}
}

$\mathrm{can}(I)\gets I$ \Comment{We obtain the canonical configuration.}

$J \gets L(J)$ \Comment{The same procedure as above goes from $J$.}

\While{$T_{B \to A}(J) \not= J$}{
    $J \gets T_{B \to A}(J)$
}

$\mathrm{can}(J)\gets J$

\Return{$\mathrm{can}(I)=\mathrm{can}(J)$}
\end{algorithm}

\paragraph{Proof of correctness.} First, we show that the left-normalisation procedure always successfully moves all $t$ tokens on a chain to the first $t$ leftmost vertices of that same chain while preserving independence.

\begin{lemma}
\label{lem:left-normalisation}
Let $X \in \{A,B\}$. If $\nabla_I(X)=t$, then $L(I) \cap X= \bigcup_{i=1}^{t} \{x_i\}$.
\end{lemma}
\begin{proof}
Suppose not; then some vertex $x_i$ where $i\le t$ is unoccupied, while some $x_j$ with $j>t$ is occupied. Since $x_i <_X x_j$, we have $\Gamma(x_i)\subseteq\Gamma[x_j]$. Moving the token from $x_j$ to $x_i$ therefore preserves independence, since, otherwise, there would exist an occupied vertex $z \in \Gamma(x_i)$ such that $z \notin \Gamma[x_j]$, which is impossible. Hence, a leftward move is possible, which is a contradiction.
\end{proof}

\begin{observation}
    \label{lem:transfer-pres-norm}
    Let $I=L(I)$ be a left-normalised independent set, then $T_{B \to A}(I)=L(T_{B \to A}(I))$.
\end{observation}
\begin{proof}
    If no token can be transferred in $I$, we then have $T_{B \to A}(I)=I$ and the claim clearly follows as $I=L(I)$, so assume otherwise. Suppose a transfer moves the rightmost token of $B$ from $x$ onto a vertex $v \in A$ such that there exists an unoccupied vertex $u <_A v$. Since $\Gamma(u) \subseteq \Gamma[v]$, moving the token from $v$ to $u$ preserves independence. However, this implies $v$ is not the leftmost unoccupied vertex $z \in A$ not adjacent to an occupied vertex in $V(G) \setminus \{x\}$, a contradiction.
\end{proof}

\noindent In light of Lemma~\ref{lem:left-normalisation}, the corollary below follows immediately.

\begin{corollary}
    \label{cor:nabla-equal}
    For any $k$-vertex independent sets $I,J$, if $\nabla_I(A)=\nabla_J(A)$, then $L(I)=L(J)$.
\end{corollary}

\noindent Corollary~\ref{cor:nabla-equal} implies that $\text{can}(I)$ is well-defined and unique in each connected component of the reconfiguration graph. In the following lemma, we show that, after termination of transfers, we obtain the canonical configuration.

\begin{lemma}\label{lem:greedy}
Let $I=L(I)$ be a left-normalised independent set. If there exists an independent set $J$ such that $I\leftrightsquigarrow J\text{ and }\nabla_J(A)>\nabla_I(A),$ then $T_{B \to A}(I) \not= I$.
\end{lemma}

\begin{proof}
Obviously, we have $I \cap B \not= \emptyset$. Consider one independent set reachable from $I$, say $I' \in \mathcal{C}(I)$, satisfying $\nabla_{I'}(A)=\nabla_I(A)+1.$ Such a configuration must exist given that any path $I,I_1,I_2,\dots,I_m$ in $\mathcal{R}_k(G)$ where $I_m=J$ must include some $i \in [m]$ such that $\nabla_{I_i}(A)=\nabla_I(A)+1$. 

Now consider $L(I')$. Since we have $I \leftrightsquigarrow I' \leftrightsquigarrow L(I')$, it then follows that we have $I \leftrightsquigarrow L(I')$ by transitivity. By Lemma~\ref{lem:left-normalisation}, $I \cap A= \{a_1,a_2, \dots, a_{\nabla_I(A)}\}$, $I \cap B= \{b_1,b_2, \dots, b_{k-\nabla_I(A)}\}$, $L(I') \cap A= \{a_1,a_2, \dots, a_{\nabla_{I}(A)+1}\}$, and $L(I') \cap B= \{b_1,b_2, \dots, b_{k-\nabla_{I}(A)-1}\}$. Hence, $I \triangle L(I')= \{b_{k-\nabla_I(A)},a_{\nabla_I(A)+1}\}$. Since $L(I')$ is an independent set, $T_{B \to A}(I)=(I \setminus \{b_{k-\nabla_I(A)}\}) \sqcup \{a_{\nabla_I(A)+1}\}=L(I')$. Therefore, $T_{B \to A}(I) \not= I$ and the lemma holds.
\end{proof}

\begin{lemma}\label{lem:max}
For a given independent set $I$, the algorithm correctly computes $\emph{\text{can}}(I)$.
\end{lemma}

\begin{proof}
Suppose, for the sake of contradiction, that the algorithm terminates at $I'$ immediately after transfers, even though there exists an independent set $J$ with $J\leftrightsquigarrow I\text{ and }\nabla_J(A)>\nabla_{I'}(A).$

Since the initial configuration $I$ was first left-normalised and only transfers occurred up to $I'$, $I' = L(I')$ by Observation~\ref{lem:transfer-pres-norm}. By Lemma~\ref{lem:greedy}, the existence of $J$ implies we then have $T_{B \to A}(I') \not= I'$, which contradicts the definition of the algorithm. Thus $I'$ maximises the number of tokens on $A$ within $\mathcal{C}(I)$. Since we established $\nabla_{I'}(A)=\nabla_{\text{can}(I)}(A)$, it remains to show that $I'=\text{can}(I)$. First, observe that $\text{can}(I)=L(\text{can}(I))$ by definition. Since $I'=L(I')$, we then have $I'=L(I')=L(\text{can}(I))=\text{can}(I)$ by Corollary~\ref{cor:nabla-equal}.
\end{proof}

\noindent The correctness of the algorithm follows immediately.

\paragraph{Complexity.} Note that the algorithm clearly runs in polynomial time. Indeed, left-normalisation can be done in $O(n)$ time by counting the tokens in each chain, and each transfer can be na\"ively implemented to take $O(n)$ time. Since there are at most $k$ transfers, the total running time is $O(nk)$. We thus obtain the following.

\tjmain*

\paragraph{Parameterised complexity.} Interestingly, the ``flushing'' idea extends to graphs of arbitrary Dilworth number. Indeed, once a partition into neighbourhood-inclusion chains is fixed, a left-normalised configuration is uniquely determined by the number of tokens in each chain. We therefore represent these configurations by vectors of token counts and construct a graph whose edges correspond to legal token jumps between them. Since there are at most $(k+1)^d$ such vectors, where $d$ is the number of chains, this graph can be constructed and its connectivity tested in $n^{O(d)}$ time, which gives the following.

\xpmain*

\begin{proof}
Let $V(G)=X_1\sqcup\cdots\sqcup X_d$ be a partition into neighbourhood-inclusion chains, where $x_{i,1}<_{X_i}\cdots<_{X_i}x_{i,|X_i|}$. Such a partition can be computed in polynomial time using, for example, the algorithm from Fulkerson \cite{fulkerson1956note}.

Let $I,J$ be independent sets of size $k$. For any independent set $S$, define $\boldsymbol{\nabla}(S)=(\nabla_S(X_1),\dots,\nabla_S(X_d))$ to be a vector representing the number of tokens in each chain. Observe that the proof of Lemma~\ref{lem:left-normalisation} not only applies to two chains, but to any number: indeed, $S\leftrightsquigarrow L(S)$ and $L(S)\cap X_i=\{x_{i,1},\dots,x_{i,\nabla_S(X_i)}\}$ for every $i \in [d]$; so we then have that $\boldsymbol{\nabla}(S)= \boldsymbol{\nabla}(S')$ implies $L(S)=L(S')$.

Now define a graph $\mathcal{H}_k$ whose vertices are the left-normalised independent sets of size $k$, that is, $V(\mathcal{H}_k)=\{L(S):S \in \mathcal{R}_k(G) \}$. Add an edge between two vertices $S,S' \in V(\mathcal{H}_k)$ whenever there exist distinct $i,j$ such that $\nabla_{S'}(X_i)=\nabla_S(X_i)-1$, $\nabla_{S'}(X_j)=\nabla_S(X_j)+1$, and $\nabla_{S'}(X_p)=\nabla_S(X_p)$ for every $p\notin\{i,j\}$. In other words, $S'$ is then obtained from $S$ by jumping the token at $x_{i,\nabla_S(X_i)}$ to $x_{j,\nabla_S(X_j)+1}$. Thus every edge of $\mathcal{H}_k$ is a valid token jump and, conversely, if $S\rightarrow S'$ is a valid token jump, then either $L(S)=L(S')$, when the jump stays within one chain, or $L(S)L(S')\in E(\mathcal{H}_k)$. It clearly follows from this that $I\leftrightsquigarrow J$ if and only if $L(I)$ and $L(J)$ belong to the same connected component of $\mathcal{H}_k$.

Observe that every vertex of $\mathcal{H}_k$ is uniquely determined by its vector $\boldsymbol{\nabla}(S)$. To bound the number of vertices, count that there are at most $\prod_{i=1}^d(\min\{|X_i|,k\}+1)\leq(k+1)^d$ candidate vectors. For each such candidate vector $\mathbf{c}=(c_1,\dots,c_d)$ with $0\leq c_i\leq\min\{|X_i|,k\}$ and $\sum_{i=1}^d c_i=k$, we test whether $\bigcup_{i=1}^d\{x_{i,1},\dots,x_{i,c_i}\}$ is an independent set, which can be done in $O(n^2)$ time. We then keep precisely those vectors for which this test succeeds. For each pair of vectors we keep, adjacency in $\mathcal{H}_k$ can be easily checked in $O(d)$ time. Since $k\leq n$, constructing $\mathcal{H}_k$ and testing whether $L(I)$ and $L(J)$ belong to the same connected component takes $n^{O(d)}$ time.
\end{proof}

%% file: st.tex
For \textsc{Sliding Token}, moves are more constrained than in \textsc{Token Jumping}, since a token may only move along an edge of the graph. As a consequence, the mechanism presented in the previous section cannot be applied. However, the neighbourhood-inclusion structure of threshold signed graphs still allows us to define a canonical configuration that will then be used to decide reachability. Unsurprisingly, the correctness proofs are slightly more involved.

\paragraph{Section-specific terminology I.} Let $G=(V,E)$ be a threshold signed graph, and let $V(G)=A\sqcup B$ be a fixed partition into two neighbourhood-inclusion chains, where $a_1 <_A a_2 <_A\cdots <_A a_{|A|}$ and $b_1 <_B b_2 <_B \cdots <_B b_{|B|}.$ For a chain $X \in \{A,B\}$ and $u,v \in X$, we say that $u$ is on the \emph{left} of $v$ when $u <_X v$; symmetrically, $v$ is on the \emph{right} of $u$. Denote by $G_v$ the maximal connected component containing $v$ in $G$ and define $d_G(u,v)$ to be the length of a+ shortest path between $u$ and $v$ in $G$, where $d_G(u,v)=\infty$ if $G_u \not=G_v$. For a token occupying a vertex $v\in I$, let $G_v^I$ and $M_I(v)$ be defined as: \[G_v^I=G_v\Bigg[V(G_v) \setminus \bigcup_{u \in I \setminus \{v\}} \Gamma[u]\Bigg] \quad \text{ and } \quad M_I(v)=\{u:u \in V(G_v^I) \land d_{G^I_v}(u,v) \not= \infty\}.\] Equivalently, $u\in M_I(v)$ if and only if there exists a path $v=v_0,v_1,\dots,v_t=u$ such that \[ \Gamma[v_i]\cap (I\setminus\{v\})=\emptyset \qquad\text{for every }0\leq i\leq t. \] In other words, the token on $v$ can be slid onto $u$ without moving any other token. Note that $v \in M_I(v)$ and all vertices in $M_I(v) \setminus \{v\}$ are unoccupied. Furthermore, notice the following, which follows from the reversibility of moves:

\begin{observation}\label{obs:M-invariant} For $u\in M_I(v)$, let $I_u=(I\setminus\{v\})\cup\{u\}$. Then, $M_{I_u}(u)=M_I(v).$ \end{observation}

\noindent A \emph{left-normalisation} of a configuration $I$ consists of repeatedly performing the following as long as possible: Choose a chain $X \in \{A,B\}$ where $|X \cap I| \geq 1$ and let $v=\max_{<_X}(I \cap X)$ be its rightmost occupied vertex. Then, slide the token on $v$ onto $u=\min_{<_X}(M_I(v) \cap X)$ along a path $v=v_0,v_1,v_2,\dots,v_t=u$ in $G$ satisfying $u <_X v$ and $\Gamma[v_i] \cap (I \setminus \{v\}) = \emptyset$ for all $i \in [t]$; i.e.~without moving any other token. Such a path exists by the definition of $M_I(v)$. If both chains admit such a move, perform it on $A$. Let $I \gets (I \setminus \{v\}) \cup \{u\}$, and repeat until no such move can be performed on either chain. Define $L(I)$ to be the configuration resulting from the left-normalisation of $I$. Since the selection of the chain and the destination of each move are uniquely determined, $L(I)$ is well-defined.

For distinct $X,Y \in \{A,B\}$, a \emph{transfer} from $X$ to $Y$ of a configuration $I$, where $|X \cap I| \geq 1$ and $M_I(v) \cap Y \not= \emptyset$, consists of sliding the token on $v=\max_{<_X}(I \cap X)$ onto $u=\min_{<_Y}(M_I(v) \cap Y)$ along a path $v=v_0,v_1,v_2,\dots,v_t=u$ in $G$ satisfying $\Gamma[v_i] \cap (I \setminus \{v\}) = \emptyset$ for all $i \in [t]$. Informally, it consists of sliding the rightmost token of $X$ onto the leftmost unoccupied vertex in $Y$ reachable without moving any other token. Define $T_{X \to Y}(I)=(I \setminus \{v\}) \cup \{u\}$ to be the resulting configuration. For convenience, if no such transfer is possible in $I$, let $T_{X \to Y}(I)=I$. Again, since the destination is uniquely determined, $T_{X \to Y}(I)$ is well-defined.

\paragraph{Canonical configuration.} Assign weights to the vertices by setting $w(a_i)=2^{n-i+1}$ and $w(b_i)=2^{|B|-i+1}$. We define the canonical configuration $\text{can}(I)$ as follows: \[ \mathrm{can}(I)=\operatorname*{arg\,max}_{I\leftrightsquigarrow J} W(J), \] where $W(I)=\sum_{v\in I}w(v)$. Observe that since the weights are powers of $2$, distinct configurations have distinct weights, and hence the canonical configuration is unique and well-defined for each connected component of the reconfiguration graph.

\paragraph{Algorithm.} Starting from $S$, the algorithm repeatedly performs a full pass as follows. First, it left-normalises the current configuration. It then performs transfers from $A$ to $B$ as long as possible, left-normalising after each transfer. Once no further transfer from $A$ to $B$ is possible, it performs transfers from $B$ to $A$ in the same way, again left-normalising after each transfer. The algorithm terminates when the configuration at the end of a full pass is equal to the configuration at its beginning, and returns the resulting configuration. We prove below that this configuration is $\text{can}(S)$.


\begin{algorithm} \DontPrintSemicolon \caption{\textsc{Sliding-Token-Reachability}} \KwIn{A threshold signed graph $G=(V,E)$ and two independent sets $I$ and $J$.}

\SetKwFunction{Canonical}{Canonical} 
\SetKwProg{Fn}{Function}{}{} 
\Fn{\Canonical{$S$}}{ 
\Repeat{$S_{\text{start}}=S$}{ $S_{\text{start}} \gets S$\; $S \gets L(S)$\; \While{$T_{A \to B}(S) \not= S$}{ $S \gets L(T_{A \to B}(S))$\; } \While{$T_{B \to A}(S) \not= S$}{ $S \gets L(T_{B \to A}(S))$\; } } \Return{$S$}\; } $\mathrm{can}(I)\gets\Canonical(I)$\; $\mathrm{can}(J)\gets\Canonical(J)$\; \Return{$\mathrm{can}(I)=\mathrm{can}(J)$}\; \end{algorithm}

\paragraph{Proof of correctness.} To give a bird's-eye view before diving into the technicalities, the proof proceeds in three steps. First, we establish some structural restrictions imposed by the two neighbourhood-inclusion chains. In particular, only the rightmost token of each chain can slide, and tokens also preserve their relative order. Second, we show that, between the beginning and the end of a full loop of \texttt{Canonical}, the tokens are always placed on vertices whose weights are not lower. As a consequence, every loop that changes the configuration strictly increases its total weight computed by $W$. Finally, we prove that a configuration left unchanged by a full loop must be canonical. Indeed, this can be shown by contradiction: Suppose otherwise, and consider a reconfiguration sequence to a configuration of larger total weight. Relative to the starting configuration, consider the first move that places a token on a vertex of strictly greater weight than its initial vertex. We show that the existence of such a move implies that a full loop of \texttt{Canonical} would change the starting configuration, contradicting our assumption.

In what follows, we show that in any configuration, only the rightmost token of a chain can be slid, and when a token is slid onto another chain, it becomes its rightmost token.

\begin{lemma}\label{lem:rightmost} Let $I,I'$ be adjacent configurations in $\mathcal{R}_k(G)$, and write $I\setminus I'=\{v\}$. If $v\in X$, where $X\in\{A,B\}$, then $v=\max_{<_X}(I\cap X)$. \end{lemma}

\begin{proof} Let $I'\setminus I=\{u\}$. Since the move is a token slide, $u\in\Gamma(v)$. Suppose there exists $w\in I\cap X$ with $v<_X w$. Then $\Gamma(v)\subseteq\Gamma[w]$, so $u\in\Gamma[w]$. Since $u\notin I$, we have $u\neq w$, and hence $\{u,w\}\in E(G)$. But we then have $u,w \in I'$, which is a contradiction given that $I'$ is an independent set. \end{proof}

\begin{corollary} \label{cor:atmosttwo} Let $I,I'$ be adjacent configurations in $\mathcal{R}_k(G)$, and write $I'\setminus I=\{u\}$. If $u\in X$, where $X\in\{A,B\}$, then $u=\max_{<_X}(I'\cap X)$. \end{corollary} \begin{proof} Simply apply Lemma~\ref{lem:rightmost} to the reverse slide from $I'$ to $I$. \end{proof}

\noindent As previously mentioned, the above observation reveals an elegant invariant: At any point in the reconfiguration, only the rightmost token of each chain can be moved. Furthermore, after a token is transferred into a chain, that token becomes the unique movable token in its new chain. Thus, the set of possible token slides is highly constrained. In particular, the vertices admit a natural linear order given by $a_1,a_2,\dots,a_{|A|},b_{|B|},b_{|B|-1},\dots,b_1$, under which at most two tokens are movable at any time and no token can overtake another token within this order.\footnote{More generally, for any incomparability graph, sliding a token preserves
the relative order of tokens in any fixed linear extension of a
poset representing the graph.} Observe that, beyond threshold signed graphs, this property is true for any arbitrary graph $H$: for any configuration of $H$, only at most $\mathcal{D}(H)$ tokens are movable, and these tokens are always the rightmost of their respective chain. We will come back to this observation in the conclusion.

In the proofs that follow, we show that moving a token to the right within its chain cannot enable a slide of another token that was previously impossible.

\begin{lemma}\label{lem:block} Let $I_1$ be an independent set, and let $u\in I_1$ and $v\notin I_1$ belong to the same chain with $u <_X v$. Suppose that $I_2=(I_1\setminus\{u\})\cup\{v\}$ is independent. Then \[ M_{I_2}(x)\subseteq M_{I_1}(x) \qquad\text{for every }x\in I_1\setminus\{u\}. \] \end{lemma}

\begin{proof} Fix $x\in I_1\setminus\{u\}$ ($x \in I_2$ since $I_1 \triangle I_2=\{u,v\}$) and $y\in M_{I_2}(x)$. By definition, there exists a path $x=v_0,v_1,\dots,v_t=y$ such that for all $i \in [t]$, $\Gamma[v_i] \cap (I_2 \setminus \{x\}) = \emptyset$. Call that path $P$. Since $v \in I_2$, we have $w \notin P$ for all $w \in \Gamma[v]$. Furthermore, $\Gamma(u) \subseteq \Gamma[v]$ by $u <_X v$, hence $w \notin P$ for all $w \in \Gamma(u)$. We also have $u \notin P$ since $u \not=\max_{<_X}(I_2\cap X)$ by Corollary~\ref{cor:atmosttwo}; otherwise, since $x\neq u$, some slide along $P$ would enter $u$ while $v$ remains occupied. As $u<_Xv$, this contradicts Corollary~\ref{cor:atmosttwo}. Therefore, since $I_1 \triangle I_2=\{u,v\}$, $\Gamma[w] \cap (I_1 \setminus \{x\}) = \emptyset$ for all $w \in P$, thus $y \in M_{I_1}(x)$ and the lemma follows. \end{proof}

\begin{lemma}\label{lem:greedy-transfer} Let $I$ be an independent set, let $X,Y\in\{A,B\}$ be distinct chains, and let $v=\max_{<_X}(I\cap X)$. Suppose that $M_I(v)\cap Y\neq\emptyset$, and put $u=\min_{<_Y}(M_I(v)\cap Y)$. For any $w\in M_I(v)\cap Y$, define $I_u=(I\setminus\{v\})\cup\{u\}$ and $I_w=(I\setminus\{v\})\cup\{w\}.$ Then \[ M_{I_w}(x)\subseteq M_{I_u}(x) \qquad\text{for every }x\in I\setminus\{v\}. \] \end{lemma}

\begin{proof} Both $I_u$ and $I_w$ are independent by the definition of $M_I(v)$. If $u=w$, there is nothing to prove. Otherwise, $u<_Y w$, and the claim follows from Lemma~\ref{lem:block}, since $I_w=(I_u\setminus\{u\})\cup\{w\}$. \end{proof}

\noindent The previous claims motivate the two operations used by our algorithm: left-normalisation and transfers between chains. Moving a token further to the left within its chain cannot reduce the reachable set of any other token while the remaining tokens are fixed. Similarly, when moving a token from one chain to another while keeping all other tokens fixed, choosing the leftmost reachable destination preserves every single-token movement available to the other tokens after any alternative choice of destination on that same chain.

\paragraph{Section-specific terminology II.} To analyse the algorithm, we introduce the following terminology. We write $\pi_i(S)$ for the vertex occupied by token $i$ in $S$. An $A\!\to\!B$ \emph{half-pass} repeatedly performs a transfer from $A$ to $B$, followed immediately by left-normalisation, until no further such transfer is possible. Thus, each successful step replaces the current configuration $S$ by $L(T_{A\to B}(S))$. A $B\!\to\!A$ half-pass is defined symmetrically. A \emph{pass} first left-normalises the current configuration, then performs an $A\!\to\!B$ half-pass followed by a $B\!\to\!A$ half-pass. Define $P(I)$ and $H_{X \to Y}(I)$ to be the configurations resulting from a pass and an $X\!\to\!Y$ half-pass starting from $I$, respectively. An \emph{execution sequence} of a pass or half-pass from $I$ is the sequence $I=I_0,I_1,\dots,I_r$ obtained by recording the initial configuration and the configuration after each completed left-normalisation or transfer. Intermediate configurations along the slide paths which implement these operations are omitted; that is, we have $I_{i}=T_{X \to Y}(I_{i-1})$ or $I_{i}=L(I_{i-1})$ for $X \in \{A,B\}$ and all $i \in [r]$. The final configuration is $P(I)$ or $H_{X\to Y}(I)$, respectively. A token occupying $v \in I$ is $I$-\emph{exposed} if $v=\max_{<_X}(I \cap X)$ for $X \in \{A,B\}$. A pass or half-pass \emph{exposes} a token if that token is exposed in at least one configuration of its execution sequence.

Fix a reference configuration $I$ and label each token by the vertex it occupies in $I$. A token initially at $x\in I$ is \emph{improved relative to $I$} in a configuration $I'$ if its position $y$ in $I'$ satisfies $w(y)>w(x)$. Equivalently, either $x,y$ belong to the same chain $X$ and $y<_X x$, or $x\in B$ and $y\in A$. A configuration $I'$ is \emph{$I$-improving} if at least one token is improved relative to $I$.

In the next lemma, we formally establish the easy observation that moving a token as far to the left in another chain always maximises the weight of the current configuration.

\begin{lemma}\label{lem:transfer-weight} Let $I$ be an independent set and let $X,Y\in\{A,B\}$ be distinct. Suppose $I\cap X\neq\emptyset$, and put $v=\max_{<_X}(I\cap X)$. If $M_I(v)\cap Y\neq\emptyset$, then, for every $u\in M_I(v)\cap Y$, \[ W\bigl(T_{X\to Y}(I)\bigr) \geq W\bigl((I\setminus\{v\})\cup\{u\}\bigr), \] with equality if and only if $u=\min_{<_Y}(M_I(v)\cap Y)$. \end{lemma}

\begin{proof} Let $z=\min_{<_Y}(M_I(v)\cap Y)$. Since vertex weights strictly decrease along $Y$, $w(z)\geq w(u)$, with equality if and only if $z=u$. Consequently, $W\bigl(T_{X\to Y}(I)\bigr) = W(I)-w(v)+w(z) \geq W(I)-w(v)+w(u),$ so the lemma holds. \end{proof}

\noindent Our next lemma shows that the set of vertices reachable by an exposed token, with all other tokens fixed, is actually only determined by at most two other tokens: the token immediately to its left in its chain and the rightmost token of the other chain, when they exist. Moreover, if these two neighbouring tokens occupy vertices no farther right in their respective chains in another independent set, while keeping the exposed token at the same vertex, its reachable set still cannot decrease. This property will come in handy in the proof of Theorem~\ref{thm:can} to show that the successive slides of a token in a hypothetical reconfiguration sequence remain possible with the other tokens fixed at their positions in a particular configuration of the pass.

\begin{lemma}
\label{claim:two-blockers} 
Let $I,I'$ be independent sets and let $X,Y\in\{A,B\}$ be distinct. Suppose that, for each $S\in\{I,I'\}$, $|S\cap X|\geq2$, $|S\cap Y|\geq1$, and $v=\max_{<_X}(S\cap X)$. Let $l=\max_{<_X}((I\setminus\{v\})\cap X)$, $r=\max_{<_Y}(I\cap Y)$, $l'=\max_{<_X}((I'\setminus\{v\})\cap X)$, and $r'=\max_{<_Y}(I'\cap Y)$. If $l'\leq_X l$ and $r'\leq_Y r$, then $M_I(v)\subseteq M_{I'}(v)$. \end{lemma} \begin{proof} Let $J=\{v,l,r\}$ and $J'=\{v,l',r'\}$. For every $z\in I\setminus J$, either $\Gamma(z)\subseteq\Gamma[l]$ or $\Gamma(z)\subseteq\Gamma[r]$. Thus $M_J(v)$ contains no neighbour of $z$ and cannot contain $z$ itself, since $z\neq v$, $z <_X w$ for some $X \in \{A,B\}$ and $w \in \{l,r\}$, which contradicts Corollary~\ref{cor:atmosttwo}. Hence $M_I(v)=M_J(v)$. Similarly, $M_{I'}(v)=M_{J'}(v)$. Applying Lemma~\ref{lem:block} first to $l'\leq_X l$ and then to $r'\leq_Y r$ gives $M_J(v)\subseteq M_{J'}(v)$, and the claim holds. 
\end{proof}

\noindent Lemma~\ref{claim:two-blockers} also holds when $I\cap X=I'\cap X=\{v\}$ or $I\cap Y=I'\cap Y=\emptyset$. In these cases, omit $l,l'$ or $r,r'$, respectively, together with the corresponding inequality. This does not affect the proof.

\begin{lemma}\label{lem:pass-monotonicity} Let $I'=P(I)$. For each $u\in I$, let $v\in I'$ be the vertex occupied by the token initially at $u$. Then: \begin{itemize} 
    \item If $u\in A$, then $v\in A$ and $v\leq_A u$. 
    \item If $u\in B$, then either $v\in A$, or $v\in B$ and $v\leq_B u$. 
\end{itemize} 
Here, $v\leq_X u$ means $v=u$ or $v<_X u$. 
\end{lemma}

\begin{proof} 
Label the tokens $1,\dots,k$ according to their order on $a_1,\dots,a_{|A|},b_{|B|},\dots,b_1$. By Lemma~\ref{lem:rightmost} and Corollary~\ref{cor:atmosttwo}, this order is invariant. Let $m=|I\cap A|$, and suppose that $H_{A\to B}(L(I))$ transfers tokens $m,\dots,\ell+1$. For $\ell<i\leq m$, let $d_i$ be the vertex occupied by token $i$ immediately before its transfer; thus $d_i\leq_A \pi_i(I)$. Set $D_0=H_{A\to B}(L(I))$ and $D_q=(L\circ T_{B\to A})^q(D_0)$.\footnote{We write $f^q(x)=(\underbrace{f \circ f \circ \dots \circ f}_{q\text{ times}})(x)$.} Let $S_0,\dots,S_r$ be the execution sequence of $H_{A\to B}(L(I))$. For $\ell<t\leq m$, define $h_t=\max\{h\in\{0,\dots,r\}: \pi_t(S_h)\in A\}$ and $C_t=S_{h_t}$. In other words, $C_t$ is the last configuration of the sequence where token $t$ occupies a vertex in $A$; thus, we have $\pi_t(C_t)=d_t$. We prove by induction on $q=0,\dots,m-\ell$ that $|D_q\cap A|=\ell+q$ and that $\pi_j(D_q)\in A$ and $\pi_j(D_q)\leq_A d_j$ for every $\ell<j\leq\ell+q$. For $q=0$, there is nothing to prove since the equality follows by definition and the condition on $j$ is vacuous anyway. If $m=\ell$, we are done. So consider $q=1$ and let $E=T_{A\to B}(C_{\ell+1})$.

By reversibility, we get $d_{\ell+1}\in M_E(\pi_{\ell+1}(E))$. Since this is the last $A\to B$ transfer, $D_0=L(E)$. During this left-normalisation, only tokens $\ell$ (if $\ell>0$) and $\ell+1$ can change position. If $\ell > 0$, a leftward move of token $\ell$ onto $z \leq_A \pi_\ell(E)$ preserves $d_{\ell+1}\in M_{S}(\pi_{\ell+1}(S))$ for $S=(E \setminus \{\pi_{\ell}(E)\}) \cup \{z\}$ by Lemma~\ref{lem:block}. A leftward move of token $\ell+1$ also preserves it: for arbitrary $S \subseteq V(G)$, if $S'=(S\setminus\{x\})\cup\{y\}$ with $x=\pi_{\ell+1}(S)$ and $y\in M_S(x)$, then clearly $M_{S'}(y)=M_S(x)$ by Observation~\ref{obs:M-invariant}. Consequently, $d_{\ell+1}\in M_{D_0}(\pi_{\ell+1}(D_0))$. Thus $T_{B\to A}(D_0)\neq D_0$, $|D_1\cap A|=\ell+1$, and $\pi_{\ell+1}(D_1)\leq_A d_{\ell+1}$.

Now let $1<q\leq m-\ell$, assume the claim for $q-1$, and write $i=\ell+q$. Let $E=T_{A\to B}(C_i)$, $F=C_{i-1}$, and $R=T_{B\to A}(D_{q-2})$; thus $D_{q-1}=L(R)$. By reversibility, $d_i\in M_E(\pi_i(E))$. Between $E$ and $F$, only left-normalisation is performed. As in the case $q=1$, Lemma~\ref{lem:block} and Observation~\ref{obs:M-invariant} therefore give $d_i\in M_F(\pi_i(F))$. Between $F$ and $R$, token $i$ is unchanged: after the transfer of token $i-1$ to $B$, token $i$ is not exposed until token $i-1$ occupies a vertex in $A$. Hence $\pi_i(R)=\pi_i(F)$. Moreover, $\pi_{i-1}(F)=d_{i-1}$ and, by the previous induction step, we have $\pi_{i-1}(R)\leq_A d_{i-1}$. If $i<k$, token $i+1$ is unchanged between $F$ and $R$, so $\pi_{i+1}(R)=\pi_{i+1}(F)$. In both $F$ and $R$, token $i-1$ is rightmost in $A$ and token $i+1$, if present, is the rightmost token of $B$ other than token $i$. Lemma~\ref{claim:two-blockers} then gives $M_F(\pi_i(F))\subseteq M_R(\pi_i(R))$. Thus $d_i\in M_R(\pi_i(R))$. Applying Lemma~\ref{lem:block} and Observation~\ref{obs:M-invariant} for $D_{q-1}=L(R)$ gives $d_i\in M_{D_{q-1}}(\pi_i(D_{q-1}))$. Therefore $T_{B\to A}(D_{q-1})\neq D_{q-1}$ and $\pi_i(T_{B\to A}(D_{q-1}))\leq_A d_i$. After left-normalisation, $|D_q\cap A|=\ell+q$ and $\pi_i(D_q)\leq_A d_i$. For $\ell<j<i$, we also have $\pi_j(D_q)\leq_A \pi_j(D_{q-1})\leq_A d_j$, which completes our induction.

Finally, now that we have established that a token occupying a vertex $v \in A$ at $I$ occupies a vertex $u \in A$ such that $u \leq_A v$ at $P(I)$, it suffices to show that a token occupying a vertex $v \in B$ at $I$ which occupies a vertex $u \in B$ at $P(I)$ satisfies $u \leq_B v$. In fact, observe that such a token (call it $t$) can only be moved by left-normalisation. Indeed, otherwise a $B\!\to\!A$ half-pass would have placed token $t$ onto $A$, and $\pi_t(P(I)) \in A$, a contradiction. By definition of left-normalisation, the claim holds, which completes the proof. 
\end{proof}

\noindent Notice that the proof of Lemma~\ref{lem:pass-monotonicity} establishes the following stronger property: For every token initially in $A$, the subsequence of vertices of $A$ occupied by that token during the execution sequence of the pass is non-increasing with respect to $\leq_A$. Likewise, for every token initially in $B$, the subsequence of vertices of $B$ occupied by that token is non-increasing with respect to $\leq_B$. And lastly, once such a token is transferred to $A$, it does not return to $B$ during the same pass. More formally:

\begin{corollary}\label{cor:fixed-pass} Let $I=\Delta_0,\Delta_1,\dots,\Delta_r=P(I)$ be the execution sequence of a pass. If $P(I)=I$, then $L(I)=I$ and, for every $i\in[k]$, $q\in\{0,\dots,r\}$, and $X\in\{A,B\}$, $\pi_i(I),\pi_i(\Delta_q)\in X$ implies $\pi_i(\Delta_q)=\pi_i(I)$. \end{corollary}

\begin{corollary}\label{cor:pass-weight} Let $I'=P(I)$. Then $W(I')\geq W(I)$, with equality if and only if $I'=I$. \end{corollary}

\noindent The corollary above and the invariant relative order of the tokens further lead us to the following crucial lemma.

\begin{lemma}\label{lem:weight-improvement} Let $I,J$ be independent sets of size $k$, with tokens labelled in their relative order. If $W(J)>W(I)$, then at least one of the following holds: \begin{enumerate} \item There exist $i\in[k]$ and $X\in\{A,B\}$ such that $\pi_i(I),\pi_i(J)\in X$ and $\pi_i(J)<_X\pi_i(I)$. \item There exists $i\in[k]$ such that $\pi_i(I)\in B$ and $\pi_i(J)\in A$. \end{enumerate} \end{lemma} \begin{proof} If neither holds, then $w(\pi_i(J))\leq w(\pi_i(I))$ all $i \in [k]$, which contradicts $W(J) > W(I)$. \end{proof}

\noindent We are now finally ready to prove that our algorithm correctly computes the canonical configuration of an arbitrary configuration.

\begin{lemma}\label{lem:forced-exposure} Let $I\neq\operatorname{can}(I)$, and let $I=S_0,S_1,\dots,S_r=\operatorname{can}(I)$ be any reconfiguration sequence. Let $h=\min\{a:S_a\text{ is }I\text{-improving}\}$, and let $t$ be the token moved by $S_{h-1}\rightarrow S_h$. If $\pi_t(I)\in X$, where $X\in\{A,B\}$, then the execution sequence of the pass from $I$ contains a configuration $D$ such that $\pi_t(D)=\max_{<_X}(D\cap X)$, that is, token $t$ is $D$-exposed in $X$. \end{lemma} 
\begin{proof} As usual, label the tokens $1,\dots,k$ according to their order on $a_1,\dots,a_{|A|},b_{|B|},\dots,b_1$. If $\pi_t(I)\in B$, no token initially in $B$ enters $A$ before $S_{h-1}\rightarrow S_h$, since such a move would be $I$-improving by definition. Therefore, token $t$ is $I$-exposed in $B$ and the claim trivially follows. So suppose we have $\pi_t(I) \in A$ instead and consider the following claim.

\begin{claim} \label{claim:exp-A} Let $n_A=|I\cap A|$ and $D_s=(L\circ T_{A\to B})^s(L(I))$. For $0\leq s\leq n_A-t$, we have $|D_s\cap A|=n_A-s$ and, for every $n_A-s<p\leq k$ and $0\leq a<h$, $\pi_p(S_a)\in B$ implies $\pi_p(D_s)\leq_B \pi_p(S_a)$. \end{claim}

\begin{proof} 
We prove this by induction on $s=0,1,\dots,n_A-t$. Write $a_p=\pi_p(I)$. For $s=0$, $|D_0\cap A|=n_A$ and $\pi_p(D_0)\leq_B a_p$ for $p>n_A$. Since no configuration $S_a$ with $a<h$ is $I$-improving, $\pi_p(S_a)\in B$ implies $a_p\leq_B \pi_p(S_a)$. Thus the invariant holds for $s=0$. Now, suppose it holds for $s<n_A-t$, and let $j=n_A-s$. Token $j$ is $D_s$-exposed in $A$. Before first becoming exposed during the half-pass, it stays at $a_j$; subsequent leftward moves preserve $a_j\in M_{D_s}(\pi_j(D_s))$ by Lemma~\ref{lem:block} and Observation~\ref{obs:M-invariant}. Whenever token $j$ slides in $S_0,\dots,S_{h-1}$, token $j-1$ belongs to $A$ and satisfies $\pi_{j-1}(D_s) = a_{j-1}\leq_A \pi_{j-1}(S_a)$. Token $j+1$ belongs to $B$ and satisfies $\pi_{j+1}(D_s)\leq_B \pi_{j+1}(S_a)$ by the induction hypothesis. 

Starting from $a_j\in M_{D_s}(\pi_j(D_s))$, let us step aside for a moment and prove by induction on $0\leq a<h$ that $\pi_j(S_a)\in M_{D_s}(\pi_j(D_s))$. The base case is trivial since we already have $\pi_j(S_0)=a_j$ by definition. For the induction step, suppose $a+1<h$. If the slide $S_a\to S_{a+1}$ does not move token $j$, the claim is unchanged; otherwise, let $Q_a=(D_s\setminus\{\pi_j(D_s)\})\cup\{\pi_j(S_a)\}$. By the induction hypothesis, $Q_a$ is an independent set, and Observation~\ref{obs:M-invariant} gives $M_{Q_a}(\pi_j(S_a))=M_{D_s}(\pi_j(D_s))$. Since token $j$ slides, it is $S_a$-exposed. The inequalities we established above imply that it is also $Q_a$-exposed and that Lemma~\ref{claim:two-blockers} applies, which gives us $M_{S_a}(\pi_j(S_a))\subseteq M_{Q_a}(\pi_j(S_a))$. Since $S_a\to S_{a+1}$ slides token $j$, we have $\pi_j(S_{a+1})\in M_{S_a}(\pi_j(S_a))$. Thus $\pi_j(S_{a+1})\in M_{D_s}(\pi_j(D_s))$, which completes the nested induction.

Now back to our original claim, note that absent neighbouring tokens (i.e.~$j-1=0$ or $j+1=k+1$) do not affect the argument in any way, so they can be safely omitted. Since $\pi_t(I)\in A$ and token $t$ is improved in $S_h$ relative to $I$, we have $\pi_t(S_h)\in A$. By Corollary~\ref{cor:atmosttwo}, token $t$ is $S_h$-exposed in $A$. Hence, for every $j>t$, $\pi_j(S_{h-1})=\pi_j(S_h)\in B$. Consequently, $T_{A\to B}(D_s) \not= D_s$, which gives, by its definition (as well as the left-normalisation) $\pi_j(D_{s+1})\leq_B \pi_j(S_a)$ whenever $a<h$ and $\pi_j(S_a)\in B$. Tokens $p>j$ remain fixed during this transfer and the subsequent left-normalisation, since they are not $T_{A\to B}(D_s)$-exposed. Hence we have $|D_{s+1}\cap A|=n_A-s-1$ and all required inequalities hold, which completes our induction. \end{proof} \noindent Finally, the lemma follows from Claim~\ref{claim:exp-A} by taking $s=n_A-t$, which completes the proof. \end{proof}

\begin{theorem} \label{thm:can} If $I=P(I)$, then $I=\emph{\text{can}}(I)$. \end{theorem}

\begin{proof} For the sake of contradiction, suppose that $I=P(I)$ but $I\neq\operatorname{can}(I)$. Consider a reconfiguration sequence $I=S_0,S_1,\dots,S_r=\operatorname{can}(I)$, with tokens labelled according to their order on $a_1,\dots,a_{|A|},b_{|B|},\dots,b_1$. Since $W(\operatorname{can}(I))>W(I)$, Lemma~\ref{lem:weight-improvement} implies that $h=\min\{j:S_j\text{ is }I\text{-improving}\}$ exists. Let $E=S_{h-1}$ and $F=S_h$, and write $E\setminus F=\{u\}$ and $F\setminus E=\{v\}$. Let $t$ be the token moved by $E\rightarrow F$. Then $\pi_t(E)=u$, $\pi_t(F)=v$, and $w(v)>w(\pi_t(I))$, whereas $w(\pi_j(S_s))\leq w(\pi_j(I))$ for every token $j$ and every $0\leq s<h$. Let $a=\pi_t(I)$ and $n_A=|I\cap A|$. We now distinguish three cases.

\textit{Case 1: $a,v\in A$ and $v<_A a$.} By Lemma~\ref{lem:forced-exposure} and Claim~\ref{claim:exp-A}, token $t$ is $D=(L\circ T_{A\to B})^{n_A-t}(I)$-exposed in $A$. By $P(I)=I$, Corollary~\ref{cor:fixed-pass} and Lemma~\ref{lem:pass-monotonicity}, $\pi_t(D)=a$; moreover, $L(D)=D$. We prove $\pi_t(S_s)\in M_D(a)$ by induction on $s=0,\dots,h$. This is obviously true for $s=0$ by not moving the token at all. Now suppose it holds for $s-1$. If the slide moves a token other than $t$, then $\pi_t(S_s)=\pi_t(S_{s-1})\in M_D(a)$ by induction. Otherwise, let $x=\pi_t(S_{s-1})$, $y=\pi_t(S_s)$, and $Q=(D\setminus\{a\})\cup\{x\}$. By induction and Observation~\ref{obs:M-invariant}, $Q$ is an independent set and $M_Q(x)=M_D(a)$. Since token $t$ slides, token $t-1$, if $1 \leq t-1$, occupies a vertex in $A$ at $S_{s-1}$ and satisfies $\pi_{t-1}(D)=\pi_{t-1}(I)\leq_A \pi_{t-1}(S_{s-1})$. Token $t+1$, if $t+1 \leq k$, occupies a vertex in $B$ and satisfies $\pi_{t+1}(D)\leq_B \pi_{t+1}(S_{s-1})$ by Claim~\ref{claim:exp-A}. Lemma~\ref{lem:rightmost} and those inequalities imply that token $t$ is $Q$-exposed. Lemma~\ref{claim:two-blockers} (with absent neighbouring tokens omitted if it applies) therefore gives us $M_{S_{s-1}}(x)\subseteq M_Q(x)$. Since $y\in M_{S_{s-1}}(x)$, we obtain $y\in M_Q(x)=M_D(a)$, which completes our induction. We hence have $v=\pi_t(S_h)\in M_D(a)$, but since we also have $v<_A a$, this contradicts $L(D)=D$.

\textit{Case 2: $a\in B$ and $v\in A$.} Since no token initially in $B$ enters $A$ before $E\rightarrow F$, the invariant token order gives $t=n_A+1$. Whenever token $t$ slides before or at $E\rightarrow F$, every token $j<t$ belongs to $A$ and satisfies $\pi_j(I)\leq_A \pi_j(S_{s-1})$, while every token $j>t$ belongs to $B$ and satisfies $\pi_j(I)\leq_B \pi_j(S_{s-1})$. The induction from Case~1, with $D=I$ instead, therefore gives $v\in M_I(a)$. Since $P(I)=I$, Lemma~\ref{lem:pass-monotonicity} and Corollary~\ref{cor:fixed-pass} imply that after all tokens transferred from $A$ have returned and left-normalisation is complete, the configuration is $I$. But $v\in M_I(a)\cap A$ implies $T_{B\to A}(I)\neq I$. Thus token $t$ is transferred to $A$ and remains there until the end of the pass, which contradicts $P(I)=I$.

\textit{Case 3: $a,v\in B$ and $v<_B a$.} As in Case~2, $t=n_A+1$ and $v\in M_I(a)$. Thus token $t$ is $I$-exposed in $B$ and $v \in M_I(a)$, contradicting $L(I)=I$.\end{proof}

\noindent The correctness of the algorithm follows effortlessly. Let us now discuss the time complexity.

\paragraph{Complexity.} For a configuration $S$ and $v\in S$, the set $M_S(v)$ can be computed in $O(n^2)$ time by deleting $\bigcup_{u\in S\setminus\{v\}}\Gamma[u]$ and computing the connected component containing $v$. During a left-normalisation, only the rightmost token of each chain can change position, and these two tokens remain rightmost before and after each operation. Each completed leftward move strictly decreases the position of one of them. Hence a left-normalisation requires $O(n)$ reachable-set computations and takes $O(n^3)$ time. Each half-pass contains at most $k$ transfers, so a full pass takes $O(kn^3)$ time.

To bound the number of passes, define $\rho(a_i)=n-i+1$, $\rho(b_j)=|B|-j+1$, and $\Phi(S)=\sum_{v\in S}\rho(v)$. By Lemma~\ref{lem:pass-monotonicity}, $\Phi(P(S))\geq\Phi(S)+1$ whenever $P(S)\neq S$. Since $\Phi$ is an integer between $k$ and $kn$, there are $O(kn)$ passes that change the configuration, followed by one pass that leaves it unchanged, which is canonical by Theorem~\ref{thm:can}. So \texttt{Canonical} runs in $O(k^2n^4)$ time, assuming, of course, that we are given the ordered chain partition (which can be computed in polynomial time regardless~\cite{fulkerson1956note}). We then get the result below.

\stmain*

\paragraph{Parameterised complexity.} It remains elusive for now whether \textsc{Sliding Token} has an $\mathsf{XP}$ algorithm parameterised by the Dilworth number. The vector of token counts in chains we used for \textsc{Token Jumping} does not extend directly because, under \textsc{Sliding Token}, left-normalising two configurations with the same number of tokens in each chain need not produce the same configuration. So chain counts alone do not determine a normalised configuration. We leave this question for future work.

%% file: conclusion.tex
We introduced a new perspective on \textsc{Independent Set Reconfiguration} by exploiting the neighbourhood-inclusion chains of a graph. This led us to polynomial-time algorithms for both \textsc{Token Jumping} and \textsc{Sliding Token}, which resolves these problems for a non-trivial subclass of permutation graphs.

Our results raise several natural questions. The most immediate one is whether the chain decomposition used here can be turned into a more general algorithmic framework.

\begin{question}
Are \textsc{Token Jumping} and \textsc{Sliding Token} fixed-parameter tractable parameterised by the Dilworth number $\mathcal{D}(G)$? That is, does each problem admit an algorithm running in $f(\mathcal{D}(G))n^{O(1)}$ time for some computable function $f$?
\end{question}

\noindent \noindent On the other hand, it is conceivable that \textsc{Sliding Token} turns out to be $\mathsf{PSPACE}$-complete already on graphs with Dilworth number three. For broader classes of permutation graphs, it is unclear whether similar ordering-based arguments can capture the structure of their reconfiguration graphs. For \textsc{Sliding Token}, we observed in Section~\ref{s:st} that only the rightmost token of each chain can move at any time (which, by the way, extends to any number of chains). This suggests viewing the problem as a reconfiguration ``puzzle'' on stacks, where only the top element of each stack is accessible.

Another natural question concerns the length of the sequences produced by our algorithm. Although canonicalisation allows us to decide reachability, a sequence passing through the canonical configuration may not be shortest.

\begin{question}
Given two independent sets of a threshold signed graph that are reachable from one another under \textsc{Sliding Token}, can a shortest reconfiguration sequence between them be computed in polynomial time?
\end{question}

\noindent It would also be interesting to substantially improve the daunting $O(k^2n^4)$ running time of our \textsc{Sliding Token} algorithm. But more generally, we ask whether decompositions of permutation graphs into neighbourhood-inclusion chains can reveal nice algorithmic structure. Such decompositions may provide a route towards resolving the conjecture of Briański et al.~\cite{brianski2021reconfiguring} that \textsc{Sliding Token} is tractable on permutation graphs. Our work suggests that understanding which properties of these decompositions allow efficient reconfiguration algorithms is a promising direction for further research.

\paragraph{Acknowledgements.} We thank the anonymous reviewers of an earlier version of this paper for their helpful feedback. This work was supported by EPSRC grant EP/X039447/1.